\documentclass[12pt]{amsart}

\usepackage[utf8]{inputenc}
\usepackage{amsthm,amsmath,amsfonts,amssymb,amscd,mathrsfs,latexsym}
\usepackage{enumerate,graphicx,cases,extarrows,indentfirst,tabularx}
\usepackage{slashed}

\usepackage[english]{babel}

\newcommand{\C}{\mathbb{C}}

\newcommand{\R}{\mathbb{R}}

\newcommand{\K}{\mathbb{K}}

\newcommand{\td}{\text{d}}

\newcommand{\Res}{\text{Res}}

\newcommand{\Tr}{\text{Tr}}
\newcommand{\KK}{\mathcal{K}}

\newtheorem{theorem}{Theorem}[section]
\newtheorem{proposition}[theorem]{Proposition}
\newtheorem{lemma}[theorem]{Lemma}
\newtheorem{remark}[theorem]{Remark}
\newtheorem{definition}[theorem]{Definition}

\newtheorem{corollary}[theorem]{Corollary}

\title{Cyclic $BV_\infty$ algebra and Frobenius manifold}
\author{Hao Wen}

\newcommand{\Addresses}{{
  \bigskip
  \footnotesize

  Hao Wen, \textsc{School of Mathematical Sciences and the Key Laboratory of Pure Mathematics and Combinatorics, Nankai University, Tianjin 300071, China}\par\nopagebreak
  \textit{E-mail address} \texttt{wenhao@nankai.edu.cn}

}}

\begin{document}
\maketitle

\begin{abstract}
	We describe the construction of Frobenius manifold out of a cyclic (commutative) $BV_\infty$ algebra $(A,\Delta)$ under the assumption of a Hodge-to-de Rham degeneration property and the existence of a compatible homotopy retract of $A$ onto its cohomology.
	We then apply it to Jacobi manifolds and Hermitian manifolds, generalizing known results in literature.
\end{abstract}

\tableofcontents

\section{Introduction}

Frobenius manifolds, axiomized by Dubrobin and appeared even earlier in Kyoji Saito's work, are very important objects in mathematical physics. They can be used to formulate intriguing relations between branches of mathematics such as symplectic geometry, complex geometry, singularity theory and integrable system.
In \cite{BK}, Barannikov and Kontsevich gave a recipe of constructing Frobenius manifold structure out of a dGBV algebra with a perfect pairing and satisfying the $ddbar$-condition (see also \cite{M}).
Since then, various generalizations have been studied. 
Besides looking for more dGBV algebras with $ddbar$-condition (see for example \cite{CZ98,CZ00}), researchers also consider more general algebraic structures and more general degeneration condition.

A natural generalization of dGBV algebra is the commutative $BV_\infty$ algebra (Definition \ref{BV algebra}) structure introduced by Kravchenko in \cite{K}.
This is a graded commutative algebra equipped with a sequence of operators with relations generalizing those for dGBV  algebra.
As is expected, commutative $BV_\infty$ algebras inherit and extend many properties of dGBV algebras. In particular, it defines an $L_\infty$ structure on the underlying graded commutative algebra (see \cite{BL}).
Additional condition is needed for a nice deformation theory of the $L_\infty$ structure.
In the dGBV case, the $ddbar$-condition is sufficient.
However, it turns out to be too restrictive in general.
A natural substitute is the Hodge-to-de Rham degeneration condition motivated from the $E_1$-degeneration of Hodge-to-de Rham spectral sequence of compact K\"ahler manifolds.
An example of dGBV algebra which satisfies the latter but not the $ddbar$-condition is the one in Landau-Ginzburg theory (see \cite{LW}).
For commutative $BV_\infty$ algebra, this condition is introduced in \cite{DSV13}. See Definition \ref{HdR data} for details.

There are many examples of commutative $BV_\infty$ algebras satisfying the Hodge-to-de Rham degeneration property. See \cite{DSV15} and \cite{CW} for instance.
A homotopy hypercommutative algebra structure is an extension of Frobenius manifold structure without a pairing (see \cite{DV}).
It is proved in \cite{DV} that if a dGBV algebra satisfies the Hodge-to-de Rham degeneration condition, then there is a homotopy hypercommutative algebra structure on its cohomology.
This is generalized later in \cite{DSV13} to the case of commutative $BV_\infty$ algebra.

When a given commutative $BV_\infty$ algebra is equipped with a cyclic structure (Definition \ref{cyclic BV}), one naturally expects a genuine Frobenius manifold structure on its cohomology.
For this, one need the notion of good basis (Definition \ref{good basis}) introduced by Kyoji Saito in his deformation theory of singularities.
The existence of good basis is usually highly nontrivial. In this paper, under the assumption of Hodge-to-de Rham degeneration, we give a sufficient condition for its existence by requiring compatibility of the homotopy retract with the cyclic structure (Lemma \ref{compatibility}).
Under the degeneration and compatibility assumption, the construction of Frobenius manifold is almost routine. A slight generalization of the construction in \cite{T} can be used to show the existence to universal solution to quantum master equation (Lemma \ref{solution gamma}). The framework in \cite{L}, which is formulated for dGBV, can then be adjusted to commutative $BV_\infty$ case to construct a Frobenius manifold structure from the universal solution to quantum master equation. This leads to our main Theorem \ref{Frobenius manifold}.

In the last two section we consider the homotopy given by Hodge theory on compact oriented manifold and show they are compatible with the pairing induced by integration. Then we apply the framework described above to compact oriented Jacobi  manifolds and compact Hermitian manifolds. By \cite{DSV15} and \cite{CW}, both of them define commutative $BV_\infty$ algebra satisfying the degeneration property. We prove the existence of Frobenius manifold structure on cohomology by showing that the integration paring induces cyclic structures.
This gives a generalization of the results for Poisson manifolds and K\"ahler manifolds in \cite{CZ98,CZ00} and also a generalization of \cite{BK} to Jacobi manifolds and Hermitian manifolds different from \cite{DSV13,CW}.

\section{Commutative $BV_\infty$ algebra and degeneration condition}

In this section we give the definition of commutative $BV_\infty$ algebra and Hodge-to-de Rham degeneration condition. Some consequences of the degeneration condition are also discussed.
 
\begin{definition} \label{BV algebra}
	Let $\K$ be a field of characteristic $0$.
	A \emph{commutative $BV_\infty$ algebra}
	\begin{align*}
		(A,\cdot,\Delta_0=d,\Delta_1,\Delta_2,\cdots)
	\end{align*}
	is a unital differential graded commutative $\K$-algebra $(A,\cdot,d)$ equipped with operators $\Delta_k$ of degree $1-2k$ and of order at most $k+1$ satisfying $\Delta_k(1_A)=0$ for all $k \geq 0$ and
	\begin{align} \label{BV}
		\sum_{i=0}^k\Delta_i \Delta_{k-i} = 0 \quad \forall k\geq 0.
	\end{align}
\end{definition}

This notion of commutative $BV_\infty$ algebra was first introduced in \cite{K}, and is sometimes called derived $BV$ algebra or homotopy $BV$ algebra in literature.
Note that this is not the homotpy $BV$ algebra in full operadic sense.
When $\Delta_k = 0$ for all $k\geq 2$, it reduces to a dGBV algebra.
For small $k$, (\ref{BV}) reads:
\begin{align*}
	d^2 = 0, \quad d \Delta_1 + \Delta_1 d = 0, \quad d \Delta_2 + \Delta_1^2+ \Delta_2 d = 0.
\end{align*}
The cohomology $H(A)$ of $(A,d)$ can be viewed as a dg complex with zero differential.

Let $\hbar$ be a formal parameter of degree $2$ and $A[[\hbar]]$ be the algebra of $A$-valued formal power series in $\hbar$, then $\Delta := \sum_{k=0}^\infty \hbar^k \Delta_k$ is homogeneous of degree $1$. The identities in (\ref{BV}) can be written in a compact form as
\begin{align*}
	\Delta^2 = 0.
\end{align*}
In this way, $(A[[\hbar]], \Delta)$ is also a complex.
Note that since $\Delta_i, i\geq 1$ does not necessarily satisfy the Leibniz rule, this is in general not a dg complex.

By abuse of notation, we will also write the commutative $BV_\infty$ algebra structure on $A$ as $(A, \Delta)$.

\begin{definition}
	A homotopy retract of $(A,d)$ onto $(H(A),0)$ is the datum of chain maps
	\begin{align*}
		\iota: (H(A),0) \to (A,d),\quad p:(A,d) \to (H(A),0)
	\end{align*}
	and a contracting homotopy $h: A \to A[-1]$ such that
	\begin{align*}
		p\iota = \text{id}_{H(A)},\quad hd + dh = \iota p - \text{id}_A
	\end{align*}
	and
	\begin{align*}
		h^2=h\iota=ph=0.
	\end{align*}
\end{definition}

In literature, a homotopy retract satisfying the last three identities (also called side conditions) is said to be special. We omit the adjective for terminological simplicity.

Given a commutative $BV_\infty$ algebra $A$, one can associate to $A$ an $L_\infty$ algebra structure, see for example \cite{CW}. To have a nice deformation theory for such an $L_\infty$ algebra, one need some additional condition.
When $A$ is a dGBV algebra, the so-called $ddbar$-condition is sufficient.
It is satisfied by the dGBV algebra of polyvector fields on compact Calaibi-Yau manifold (\cite{BK}) and the dGBV algebra of Dolbeault forms on compact K\"ahler manifold (\cite{CZ00}). In the case of commutative $BV_\infty$ algebra, a natural candidate is the Hodge-to-de Rham degeneration condition.

\begin{definition} \label{HdR data}
	A unital commutative $BV_\infty$ algebra $(A,\Delta)$ is said to have a \emph{Hodge-to-de Rham degeneration data} if there is a special homotopy retract $(\iota,p,h)$ of $(A,d)$ onto $(H(A),0)$ such that
	\begin{align}\label{HdR}
		\sum_{j_1+j_2+\cdots+j_l = k;\, j_i \geq 1} p \Delta_{j_1} h \Delta_{j_2} h \cdots h \Delta_{j_l} \iota = 0 \quad \forall k\geq 1.
	\end{align}
\end{definition}

Such degeneration condition was first introduced in \cite{DV} for dGBV algebras and in \cite{DSV13} for commutative $BV_\infty$ algebra. In \cite{DSV15} the authors give several equivalent characterizations. The one given above says all the transferred operators of $\Delta_k$ on $H(A)$ by the homotoy retract $(\iota,p,h)$ vanish.
It can also be characterized by the $E_1$-degeneration of certain spectral sequence (hence the name) and a trivialization of the operator $\Delta$ on $A[[\hbar]]$. Moreover, it is shown that (\ref{HdR}) is independent of the choice of the homotopy retraction: if it holds for some particular choice of homotopy retract $(\iota,p,h)$, then it holds for any choice.
It is proved for dGBV algebras in \cite{DV} and for commutative $BV_\infty$ algebras in \cite{DSV13} that if $(A,\Delta)$ admits a Hodge-to-de Rham degeneration data, then there is a homotopy hypercommutative algebra structure on its cohomology.
The latter is an algebra over $\Omega H^\bullet (\mathcal{M}_{0,n+1})$, the operadic resolution of the operad of homology of the Deligne-Mumford moduli spaces of stable genus $0$ curves and it extends the information of a Frobenius manifold without a pairing.
\newline

In the remaining part of this paper we will always assume $A$ is a commutative $BV_\infty$ algebra admitting a Hodge-to-de Rham degeneration data.
By Proposition 2.4 of \cite{CW}, one can derive an explicit formula for the operator satisfying $d \Phi = \Phi \Delta$ and $\Phi \equiv id_A (\text{mod } \hbar)$ in terms of the homotopy retraction datum.
For our purpose, we need an explicit formula for its inverse.

\begin{lemma} \label{closed}
	Assume $(A,\Delta)$ admits a Hodge-to-de Rham degeneration data $(\iota,p,h)$, then for any $k \geq 1$,
	\begin{align*}
		d \sum_{j_1+j_2+\cdots+j_l = k;\, j_i \geq 1} \Delta_{j_1} h \cdots h \Delta_{j_l} \iota p = 0.
	\end{align*}
\end{lemma}
\begin{proof}
	We prove by induction. For $k=1$, $d \Delta_1 \iota p = - \Delta_1 d \iota p = 0$. Assume the conclusion holds for $1,2,\cdots,k-1$, then for $k$,
	\begin{align*}
		&d \sum_{j_1+j_2+\cdots+j_l = k;\, j_i \geq 1} \Delta_{j_1} h \cdots h \Delta_{j_l} \iota p \\
		=& \sum_{j=1}^k d \Delta_j \sum_{j_1+j_2+\cdots+j_l = k-j;\, j_i \geq 1} h \Delta_{j_1} h \cdots h \Delta_{j_l} \iota p \\
		=& -\sum_{j=1}^k (\sum_{s=1}^{j-1} \Delta_s \Delta_{j-s}) \sum_{j_1+j_2+\cdots+j_l = k-j;\, j_i \geq 1} h \Delta_{j_1} h \cdots h \Delta_{j_l} \iota p \\
		\quad& + \sum_{j=1}^k \Delta_j (hd + id_A-\iota p) \sum_{j_1+j_2+\cdots+j_l = k-j;\, j_i \geq 1} \Delta_{j_1} h \cdots h \Delta_{j_l} \iota p \\
		=& -\sum_{j=1}^k (\sum_{s=1}^{j-1} \Delta_s \Delta_{j-s}) \sum_{j_1+j_2+\cdots+j_l = k-j;\, j_i \geq 1} h \Delta_{j_1} h \cdots h \Delta_{j_l} \iota p \\
		\quad& + \sum_{j=1}^k \Delta_j \sum_{j_1+j_2+\cdots+j_l = k-j;\, j_i \geq 1} \Delta_{j_1} h \cdots h \Delta_{j_l} \iota p \\
		=& \,0,
	\end{align*}
	where the second last identity follows from the induction hypothesis and degeneration assumption, while the last identity follows because the terms with opposite sign cancel.
\end{proof}

\begin{proposition} \label{splitting map}
	Let $s_k := -\Delta_k h + \sum_{j_1+j_2+\cdots+j_l = k;\, j_i \geq 1} h \Delta_{j_1} h \cdots h \Delta_{j_l} \iota p$, then $S = id_A + \sum_{k\geq 1} \hbar^k s_k$ satisfies
	\begin{align*}
		\Delta S = S d.
	\end{align*}	
\end{proposition}
\begin{proof}
	Expand $S$ in powers of $\hbar$, then $\Delta S = S d$ is equivalent to a sequence of identities:
	\begin{align} \label{equation for S}
		\Delta_1 s_k + \cdots + \Delta_k s_1 + \Delta_{k+1} = s_{k+1} d - d s_{k+1} \qquad \forall k\geq 0,
	\end{align}
	which for $k=0$ should be understood as $\Delta_1 = s_1 d - d s_1$.
	Assume the identities for $0,1,\cdots,k$ hold, then we have
	\begin{align*}
		& d s_{k+1} - s_{k+1} d \\
		=& -\Delta_1 \Delta_k h + \sum_{j_1+j_2+\cdots+j_l = k;\, j_i \geq 1} \Delta_1 h \Delta_{j_1} h \cdots h \Delta_{j_l} \iota p \\
		& -\Delta_2 \Delta_{k-1} h + \sum_{j_1+j_2+\cdots+j_l = k-1;\, j_i \geq 1} \Delta_2 h \Delta_{j_1} h \cdots h \Delta_{j_l} \iota p \\
		&\cdots \\
		&-\Delta_k \Delta_1 h + \Delta_k h \Delta_1 \iota p +\Delta_{k+1} \iota p - \Delta_{k+1} dh -\Delta_{k+1} hd \\
		=&\, d \Delta_{k+1} h + \sum_{j_1+j_2+\cdots+j_l = k+1;\, j_i \geq 1} \Delta_{j_1} h \cdots h \Delta_{j_l} \iota p -\Delta_{k+1} hd
	\end{align*}
	By Lemma \ref{closed} and the degeneration assumption, we have 
	\begin{align*}
		&\sum_{j_1+j_2+\cdots+j_l = k+1;\, j_i \geq 1} \Delta_{j_1} h \cdots h \Delta_{j_l} \iota p \\ 
		=& -d \sum_{j_1+j_2+\cdots+j_l = k+1;\, j_i \geq 1} h \Delta_{j_1} h \cdots h \Delta_{j_l} \iota p
	\end{align*}
	Finally the identity $pd=0$ implies
	\begin{align*}
		 s_{k+1} := -\Delta_{k+1} h + \sum_{j_1+j_2+\cdots+j_l = k+1;\, j_i \geq 1} h \Delta_{j_1} h \cdots h \Delta_{j_l} \iota p
	\end{align*}
	solves (\ref{equation for S}).
\end{proof}

Applying $S$ to $\iota (a)$ for $a \in H(A)$ and using $h\iota = 0$, we get
\begin{corollary} \label{splitting}
	There is a map $S: H(A) \to H(A[[\hbar]],\Delta)$ given by
	\begin{align*}
		a \mapsto [\iota + \sum_{n=1}^\infty \hbar^n \sum_{j_1+j_2+\cdots+j_k = n;\, j_i \geq 1} h\Delta_{j_1}h\Delta_{j_2}\cdots h\Delta_{j_k}\iota](a).
	\end{align*}
\end{corollary}

This gives explicit formulation that if $(A,\Delta)$ satisfies the Hodge-to-de Rham degeneration property, then $H(A[[\hbar]],\Delta)$ is a free $\K[[\hbar]]$-module with
\begin{align*}
	H(A[[\hbar]],\Delta) \cong H(A)[[\hbar]],
\end{align*}
as is shown in \cite{CW}.
The results above are also generalizations of that in \cite{LW}.

\section{Cyclic structure and Frobenius manifold}

In this section we construct Frobenius manifold out of a cyclic (commutative) $BV_\infty$ algebra under the degeneration assumption and a compatibility condition.

\begin{definition} \label{cyclic BV}
	A \emph{cyclic $BV_\infty$ algebra} of dimension $n$ is a triple $(A,\Delta,\Tr)$, where $(A, \Delta)$ is a unital commutative $BV_\infty$ algebra and $\Tr$ is a homogeneous $\K$-linear map $\Tr:A \to \K[-n]$ such that the induced pairing
	\begin{align*}
		(-,-)_A: A \times A \to \K[-n], \quad (a,b)_A := \Tr(a\cdot b)
	\end{align*}
	is a $\K$-bilinear pairing satisfying
	\begin{enumerate}
		\item For any $k \geq 0$,
		\begin{align*}
			(\Delta_k a ,b ) =(-1)^{|a|+k+1}(a, \Delta_k b);
		\end{align*}
		\item The induced pairing 
		\begin{align*}
			\KK^{(0)}(-,-): H^i(A) \times H^{n-i}(A) \to \K
		\end{align*}
		given by
		\begin{align*}
			([a],[b]) \mapsto (a,b)
		\end{align*}
		is perfect, which means $H^i(A)$ is finite dimensional for all $i$ and $\KK^{(0)}$ is non-degenerate.
	\end{enumerate}
\end{definition}

Denote the $\K[[\hbar]]$-extension of $(-,-)$ by  $(-,-)_\hbar$, then by (2) of the above definition we have
\begin{align*}
	(\Delta a, b)_\hbar = (-1)^{|a|+1} (a, \Delta^- b)_\hbar \quad a,b \in A,
\end{align*}
where $\Delta^-$ is defined in the same way as $\Delta$ with $\hbar$ replaced by $-\hbar$. So there is a well-defined pairing
\begin{align*}
	\KK(-,-): H(A[[\hbar]],\Delta) \times H(A[[\hbar]],\Delta) \to \K[[\hbar]]
\end{align*}
given by
\begin{align*}
	\KK([\alpha], [\beta]) := (\alpha, \bar\beta)_\hbar \qquad \alpha,\beta \in A[[\hbar]],
\end{align*}
where $\bar\beta$ is obtained from $\beta$ by replacing $\hbar$ by $-\hbar$.

Evaluating at $\hbar=0$ gives a chain map $T: (A[[\hbar]],\Delta) \to (A,d)$ and thus a map $T$ between cohomologies.
It is obvious that
\begin{align*}
	TS=id_{H(A)}
\end{align*}
and
\begin{align*}
	\KK(\alpha,\beta)|_{\hbar = 0} = \KK^{(0)}(T(\alpha),T(\beta)).
\end{align*}
The following notion was first used by K. Saito in his study in singularity theory. 

\begin{definition} \label{good basis}
	A $\K[[\hbar]]$-basis $\alpha_1,\cdots,\alpha_\mu$ of $H(A[[\hbar]],\Delta)$ is called a \emph{good basis} if
	\begin{align*}
		\KK(\alpha_i,\alpha_j) = \KK^{(0)}(T(\alpha_i),T(\alpha_j))\quad \forall 1 \leq i,j \leq \mu,
	\end{align*}
	i.e. all the higher order terms in $\hbar$ of $\KK(\alpha_i,\alpha_j)$ vanish.
\end{definition}

It is usually difficult to prove the existence of a good basis. We give in the following lemma a sufficient condition.

\begin{lemma} \label{compatibility}
	If the map $h:A \to A[-1]$ of the homotopy retract $(\iota,p,h)$ is compatible with the cyclic structure of $(A,\Delta)$ in the sense that
	\begin{align*}
		(h (a), b) = (-1)^{|a|} (a,h (b)) \qquad \forall\, a,b \in A,
	\end{align*}
	then for any $\K$-basis $a_1,\cdots,a_\mu$ of $H(A)$, $\alpha_1 := S(a_1),\cdots,\alpha_\mu := S(a_\mu)$ form a good basis.
\end{lemma}
\begin{proof}
	By the explicit formula for $S$, it suffices to show for any $a \in H(A)$ and $b,c \in A$,
	\begin{align*}
		(\iota(a),h(b)) = (h(b),h(c))=0.
	\end{align*}
	This follows form $h^2 = h \iota = 0$ and the compatibility assumption.
\end{proof}

\begin{remark}
	Similar homotopy retract is used to transfer the structure of a cyclic $L_\infty$-algebra to its cohomology, see \cite{J} and the reference therein. 
\end{remark}

Now we outline the construction of formal Frobenius manifold structure on the formal neighborhood $\mathcal{M}$ of $0$ in $H(A)$ along the same lines as given in \cite{L}.
For this the universal solution to the quantum master equation for the $L_\infty$ structure on $A$ is important.
Instead of describing explicitly the $L_\infty$ structure or the explicit expression of the quantum master equation, we only point out that according to \cite{CJ}, $\Gamma$ is a solution to the quantum master equation if and only if 
\begin{align} \label{MC equation}
	\Delta e^{\frac{\Gamma}{\hbar}} = 0.
\end{align}
Here we view $\Delta$ as an operator on $A((\hbar))$, the algebra of $A$-valued formal Laurent series in $\hbar$.

\begin{lemma} \label{solution gamma}
	Assume as before, then there is a universal solution to the quantum master equation of the form
	\begin{align*}
		\Gamma = \sum_{i=1}^\mu \alpha_i t^i + \sum_{i,j=1}^\mu \alpha_{ij} t^i t^j + \sum_{i,j,k=1}^\mu \alpha_{ijk}t^i t^j t^k + \cdots,
	\end{align*}
	where $\alpha_1,\cdots,\alpha_\mu$ are given in Lemma \ref{compatibility}, $\alpha$'s with multiple subscript belong to $A[[\hbar]]$ and $t^i$'s are dual coordinates to $a_i$'s in $H(A)$.
\end{lemma}
\begin{proof}
	This is in fact a slight generalization of the construction in Theorem 2 of \cite{T}.
	By Proposition \ref{splitting map} and equation (\ref{MC equation}) above, we can proceed the construction by induction on powers of $t$'s.
\end{proof}

As in the case of Landau-Ginzburg theory, to a good basis we can associate a good opposite filtration.
Under the assumption in Lemma \ref{compatibility}, let 
\begin{align*}
	\mathcal{L} := \hbar^{-1} \text{Span}_\K \{\alpha_1,\cdots,\alpha_\mu\}[\hbar^{-1}],
\end{align*}
then $\mathcal{L}$ is a good opposite filtration in the sense that
\begin{enumerate}
	\item $\hbar^{-1} \mathcal{L} \subset \mathcal{L}$;
	\item $H(A((\hbar)),\Delta) = H(A[[\hbar]],\Delta) \oplus \mathcal{L}$;
	\item $\Res_{\hbar = 0} \Tr(\varphi,\psi) \td \hbar = 0, \text{ for } \forall \varphi,\psi \in \mathcal{L}$.
\end{enumerate}
See \cite{L} for details.

\begin{proposition}
	Denote by
	\begin{align*}
		\pi: H(A((\hbar)),\Delta)= H(A[[\hbar]],\Delta) \oplus \mathcal{L} \to H(A[[\hbar]],\Delta)
	\end{align*}
	the projection to the first summand, then the universal solution $\Gamma$ in Lemma \ref{solution gamma} can be written in the form
	\begin{align*}
		\Gamma = \sum_i \alpha_i \tau_\hbar^i + \sum_{ij} \alpha_{ij} \tau_\hbar^i \tau_\hbar^j + \cdots
	\end{align*}
	such that $\vec\tau :=\{\tau^i\}$ form coordinates on $\mathcal{M}$, $\tau_\hbar^i = \tau^i + O(\vec{\tau}^2) \in \K[[\hbar]][[\vec\tau]]$ and
	\begin{align*}
		\pi ([\hbar e^\frac{\Gamma}{\hbar}-\hbar]) = \sum_i \alpha_i \tau^i.
	\end{align*}
\end{proposition}

\begin{proof}
	The same proof as that of Proposition 11 in \cite{L} works. This is because $e^\frac{\Gamma}{\hbar}$ represents a class in $H(A((\hbar)),\Delta)$, as in the case of dGBV algebra.
\end{proof}

Let $\mathcal{H}_\mathcal{M}^{(0)}$ be the $\K[[\hbar]][[\vec\tau]]$-submodule of $H(A((\hbar)),\Delta)[[\vec\tau]]$ generated by $[\hbar \partial_{\tau^i} e^\frac{\Gamma}{\hbar}]$, then we have the decomposition
\begin{align*}
	H(A((\hbar)),\Delta)[[\vec\tau]] = \mathcal{H}_\mathcal{M}^{(0)} \oplus \mathcal{L}[[\vec\tau]].
\end{align*}
The map $\partial_{\tau^i} \mapsto [\hbar \partial_{\tau^i} e^\frac{\Gamma}{\hbar}]$ can be used to define a pairing and a product on $T\mathcal{M}$, which eventually leads to a formal Frobenius manifold structure on $\mathcal{M}$.
The readers can find details in Section 2.2.1 in \cite{L}.

\begin{theorem} \label{Frobenius manifold}
	Let $A$ be a commutative $BV_\infty$ algebra with a Hodge-to-de Rham degeneration data $(\iota,p,h)$. If $h$ is compatible with the cyclic structure, then there is a formal Frobenius manifold structure on the formal neighborhood $\mathcal{M}$ of $0$ in $H(A)$.
\end{theorem}

\section{Hodge theoretic homotopy retract}

In this section we show the homotopy retract given by Hodge theory is compatible with the integration pairing on a compact Riemannnian or Hermitian manifold.

Assume first that $X$ is a smooth oriented compact Riemannian manifold of real dimension $n$.
Let $A:=\Omega(X)$ be the algebra of real-valued differential forms on $X$ with Hodge grading, then we have the following well-known Hodge decomposition of operators on $\Omega(X)$:
\begin{align} \label{de Rham decomposition}
	id - \iota p  = \td \td^* G + \td^* \td G,
\end{align}
where $\iota$ is the inclusion of harmonic forms, $p$ is the harmonic projection, $\td^* := (-1)^{n(k-1)+1}*\td * : \Omega^k(X) \to \Omega^{k-1}(X)$ is the Hilbert adjoint of $\td$, $*: \Omega^k(X) \to \Omega^{n-k}(X)$ is the star operator satisfying $*^2 = (-1)^{k(n-k)}$ on $\Omega^k(X)$ and $G$ is the Green operator.
It is straightforward to verify that  $(\iota,p, h:=-\td^* G)$ is a homotopy retract from $(\Omega(X),\td)$ to $(\mathcal{H}_\td(X),0)$, where $\mathcal{H}_\td(X) \cong H(X)$ is the space of $\td$-harmonic forms.

Since $X$ is compact, $\Tr:=\int_X$ on $\Omega(X)$ satisfying for $a,b \in \Omega(X)$,
\begin{align*}
	\int_X \td a \wedge b = (-1)^{|a|+1} \int_X a \wedge \td b.
\end{align*}
Furthermore, for $a,b \in \Omega(X)$ satisfying $|a|+|b|=n+1$,
\begin{align*}
	&\int_X \td^* a \wedge b = (-1)^{|b|(n-|b|)} \int_X \td^* a \wedge *^2 b \\
	=& (-1)^{|b|(n-|b|)} \langle \td^* a ,*b \rangle = (-1)^{|b|(n-|b|)} \langle a, \td *b \rangle \\
	=& (-1)^{|b|(n-|b|)} \int_X a \wedge *\td * b  \\
	=& (-1)^{|b|(n-|b|)+n(|b|-1)+1} \int_X a \wedge \td^* b \\
	=& (-1)^{|a|} \int_X a \wedge \td^* b,
\end{align*}
where $\langle -,-\rangle$ is induced by the Riemannian metric.
Using the same argument and the fact that $G$ is self-adjoint with respect to $\langle -,-\rangle$, one can show
\begin{align*}
	\int_X h (a) \wedge b = (-1)^{|a|} \int_X a \wedge h (b).
\end{align*}
So $h=-\td^* G$ is compatible with the integration pairing.
\newline

Assume now that $X$ is a compact Hermitian manifold of real dimension $n (=2\, \text{dim}_\C X)$.
Let now $\Omega(X)$ be the algebra of Dolbeault forms on $X$ with Hodge grading, then similarly we have:
\begin{align} \label{Dolbeault decompositon}
	id - \iota p = \bar\partial \bar\partial^*G + \bar\partial^* \bar\partial G,
\end{align}
where $\bar\partial^* := -*\partial *$ is the Hilbert adjoint of $\bar\partial$ and other operators has similar meaning to those for Riemannian manifolds.
It turns out that $(\iota,p,h:= -\bar\partial^* G)$ is a homotopy retraction from $(\Omega(X), \bar\partial)$ to $(\mathcal{H}_{\bar\partial}(X),0)$, where $\mathcal{H}_{\bar\partial}(X) \cong H(X)$ is the space of $\bar\partial$-harmonic forms.

Let again $\Tr:=\int_X$ be the integration map on $X$. As pointed out in \cite{CZ00}, we have:
\begin{align*}
	\int_X \bar\partial a \wedge b =& (-1)^{|a|+1} \int_X a \wedge \bar\partial b, \\
	\int_X \partial^* a \wedge b =& (-1)^{|a|} \int_X a \wedge \partial^* b,
\end{align*}
where $\partial^*$ is the Hilbert adjoint operator of $\partial$.
By the corresponding results for $\td^*$ and $\partial^*$, for $a,b\in \Omega(X)$ satisfying $|a| =|b|+1$,
\begin{align*}
	\int_X \bar\partial^* a \wedge b = (-1)^{|a|} \int_X a \wedge \bar\partial^* b.
\end{align*}
One can argue similarly that $h = -\bar\partial^*G$ is compatible with the integration pairing.
\newline

The final ingredient we need is an identity on the contraction map.
Let $w$ be a $k$-vector field on $X$ and $w \vdash: \Omega^p(X) \to \Omega^{p-k}(X)$
be the map of contracting with $w$.
\begin{lemma} \label{contraction}
	If $w$ is a $k$-vector field and $a,b \in \Omega(X)$ satisfy $|a|+|b| = n+k$, then
	\begin{align*}
		\int_X (w \vdash a) \wedge b = (-1)^{k(|a|+1)} \int_X a \wedge (w \vdash b)
	\end{align*}
\end{lemma}
\begin{proof}
	By partition of unity, we may assume in local coordinates $\{x_i\}$, 
	\begin{align*}
		w = \sum_{i_1<\cdots< i_k} w^{i_1 \cdots i_k} \frac{\partial}{\partial x_{i_1}} \wedge \cdots \wedge \frac{\partial}{\partial x_{i_k}}.
	\end{align*}
	Using the fact that $\frac{\partial}{\partial x_i} \vdash$ is a derivation and forms of degree greater than $n=\text{dim}_\R \,X$ must vanish, we have
	\begin{align*}
		&\int_X (w \vdash a) \wedge b \\
		=& \sum_{i_1<\cdots< i_k} w^{i_1 \cdots i_k}\int_X  (\frac{\partial}{\partial x_{i_1}} \vdash \cdots \frac{\partial}{\partial x_{i_k}} \vdash a) \wedge b \\
		=& (-1)^{|a|-k} \sum_{i_1<\cdots< i_k} w^{i_1 \cdots i_k}\int_X (\frac{\partial}{\partial x_{i_2}} \vdash \cdots \frac{\partial}{\partial x_{i_k}} \vdash a) \wedge (\frac{\partial}{\partial x_{i_1}} \vdash b) \\
		=& (-1)^{k|a|-\sum_{s=1}^k s} \sum_{i_1<\cdots< i_k} w^{i_1 \cdots i_k} \int_X a \wedge (\frac{\partial}{\partial x_{i_k}} \vdash \cdots \frac{\partial}{\partial x_{i_1}} \vdash b) \\
		=& (-1)^{k|a|-\sum_{s=1}^k s+ \sum_{s=1}^{k-1} s} \int_X a \wedge (w \vdash b) \\
		=& (-1)^{k(|a|+1)} \int_X a \wedge (w \vdash b).
	\end{align*}
\end{proof}

\section{Applications}

Now we are ready to apply the framework above to some commutative $BV_\infty$ algebras of geometric origin.

Consider first the case of compact Jacobi manifolds.
By definition, a Jacobi manifold is a smooth manifold $X$ with a bi-vector field $\pi$ and a vector field $\eta$ satisfying
\begin{align*}
	[\pi,\pi] = 2\eta \pi \quad \text{and} \quad [\pi,\eta] = 0,
\end{align*}
where $[-,-]$ is the Schouten-Nijenhuis bracket. Note that Poisson manifolds correspond to the case $\eta=0$.
In \cite{DSV13} (see also \cite{DSV15}), it is shown that given a Jacobi manifold $(X,\pi,\eta)$,
\begin{align*}
	(\Omega(X), \Delta_0 := \td, \Delta_1 := [\pi \vdash, \td], \Delta_2 := \eta\pi \vdash, \Delta_{>2} := 0)
\end{align*}
is a commutative $BV_\infty$ algebra admitting a Hodge-to-de Rham degeneration data.
The following theorem removes the assumption for Poisson manifolds on isomorphism of cohomologies in Theorem 2.1 of \cite{CZ98} and generalizes the corresponding results to Jacobi manifolds.

\begin{theorem}
	Let $X$ be an oriented compact Jacobi manifold, then there exists a formal Frobenius manifold structure on the formal neighborhood of $0$ in $H(X,\R)$.
\end{theorem}
\begin{proof}
	Let $(-,-)$ be the integration pairing, we will show it induces a cyclic $BV_\infty$ structure on $\Omega(X)$. 
	Since $X$ is oriented, $H(X)$ is of finite dimension and the pairing on $H(X)$ is non-degenerate. To verify (2) of Definition \ref{cyclic BV}, note that the identify for $\Delta_0$ and $\Delta_1$ follows from the same results for Poisson manifolds, see for example \cite{CZ98}. For $\Delta_2$, it follows from the $k=3$ case of Lemma \ref{contraction} since $\eta\pi$ is a $3$-vector field. The theorem now follows from the argument in previous section and Theorem \ref{Frobenius manifold}.
\end{proof}
\begin{remark}
	Unfortunately, our construction does not work for the so-called generalized Poisson manifolds (see \cite{BL}) because integration pairing does not induce a cyclic $BV_\infty$ structure..
\end{remark}

Next we turn to the case of compact Hermitian manifolds.
An Hermitian metric on a complex manifold $X$ gives rise to a real $(1,1)$-form $\omega$ which is not necessarily closed. Let $L: \Omega^k(X) \to \Omega^{k+2}(X)$ be the Lefschetz operator defined by $L a := \omega \wedge a$ and $\Lambda:= L^*$ be its adjoint operator.
Consider the operators
\begin{align*}
	\lambda :=[\partial,L] = \partial \omega \qquad \tau:= [\Lambda,\lambda]
\end{align*}
and their adjoint operators $\lambda^*$ and $\tau^*$.

When $\omega$ is K\"{a}hler, we have $\lambda = \tau =0$. It is proved in \cite{CZ00} that in this case $(\Omega(X), \wedge, \Delta_0 = \bar\partial, \Delta_1 = -i\partial^*)$ is a dGBV algebra satisfying the $ddbar$-condition.
It is generalized in \cite{CW} that in the Hermitian case,
\begin{align*}
	(\Omega(X), \wedge, \Delta_0 = \bar\partial, \Delta_1 = -i(\partial^*+\tau^*), \Delta_2 = i \lambda^*, \Delta_{\geq 3} = 0)
\end{align*}
is a commutative $BV_\infty$ algebra admitting a Hodge-to-de Rham degeneration data.
The following theorem generalize the corresponding results for K\"ahler manifolds in \cite{CZ00}.

\begin{theorem}
	Let $X$ be a compact Hermitian manifold, then there exists a formal Frobenius manifold structure on the formal neighborhood of $0$ in $H(X)$.
\end{theorem}
\begin{proof}
	Again let $(-,-)$ be the integration pairing. 
	With in mind the argument of the previous section and comparing with the case of Jacobi manifold, it suffices to show $(-,-)$ induces a cyclic $BV_\infty$ structure on $\Omega(X)$. This amounts to verify (2) of Definition \ref{cyclic BV}. For $\Delta_0$, the identity is trivial. For $\Delta_2$, $\lambda^*$ acts as a contraction with a $3$-vector field since $\lambda$ is a $3$-form, therefore the identity follows from the $k=3$ case of Lemma \ref{contraction}. For $\Delta_1$, the identify for the term $-i \partial^*$ follows from \cite{CZ00}.
	As to the term $-i\tau^*$, we have
	\begin{align*}
		&\int_X (-i\tau^*) (a) \wedge b \\
		=& -i \int_X (\Lambda \lambda - \lambda \Lambda)^* (a) \wedge b \\
		=& -i \int_X (\lambda^* L - L \lambda^*) (a) \wedge b \\
		=& -i (-1)^{3(|a|+1)} \int_X a \wedge (L \lambda^* - \lambda^* L)(b) \\
		=& -i (-1)^{3(|a|+1)+1} \int_X a \wedge (\lambda^* L - L \lambda^*)(b) \\
		=& (-1)^{|a|}\int_X a \wedge (-i\tau^*) (b),
	\end{align*}
	where the third identity follows from the $k=3$ case of Lemma \ref{contraction} and the fact that $\omega$ is a $2$-form.
\end{proof}

Same argument as above can be applied to more commutative $BV_\infty$ algebras, though we will not attempt to do here.
In fact, we have shown that given a commutative $BV_\infty$ algebra $(A,\Delta)$ that satisfies the Hodge-to-de Rham degeneration property, if it can be equipped with a cyclic structure and the differential $\Delta_0 = d$ has a compatible Hodge theory, then the construction in this paper applies.
Thus we obtain a generalization of \cite{BK} to Jaboci manifolds and Hermitian manifolds different from that in \cite{DSV15} and \cite{CW}.
We also point out that the construction of Frobenius manifold for Landau-Ginzburg model in \cite{LW} also fit with the framework of this paper.

\section*{acknowledgements}

This work is supported by Young Scientists Fund of the National Natural Science Foundation of China (Grant No. 12201314).

\Addresses

\end{document}